\pdfoutput=1
\documentclass[12pt, draftclsnofoot, onecolumn]{IEEEtran}

\usepackage{bbm}
\usepackage{amsthm}
\usepackage{amsmath}
\usepackage{amsfonts}
\usepackage{amsmath,amsfonts,amssymb,amsbsy, amsthm,ae,aecompl}
\usepackage{algorithm,algpseudocode}
\usepackage[english]{babel}
\usepackage{bm}
\usepackage{color}
\usepackage[noadjust]{cite}
\usepackage{epsfig}
\usepackage{enumerate}
\usepackage{float} 
\usepackage{fancyhdr}
\usepackage[T1]{fontenc}
\usepackage[acronym,toc,shortcuts]{glossaries}
\usepackage{graphicx, caption, subcaption}
\usepackage{hyperref}
\usepackage{lastpage}
\usepackage{listings}
\usepackage{lipsum}
\usepackage{dsfont}
\usepackage{multirow,tabularx}
\usepackage[normalem]{ulem}
\usepackage{tikz,pgfplots}
\usepackage{times}
\usepackage{verbatim}
\usepackage[all]{xy}
\usepackage{mathtools}
\usepackage{tikz}
\usepackage{tikz-qtree,tikz-qtree-compat}
\usepackage{pgfplots}

\usepackage{lipsum}
\usetikzlibrary{calc}
\usetikzlibrary{patterns}

\usepackage{enumitem}
\graphicspath{{figures/}}

\usetikzlibrary{arrows,shapes}
\newcommand{\E}[1]{{\mathbb E}\left[ #1 \right]}

\DeclareMathOperator{\Tr}{Tr}
\newtheorem{Theorem}{Theorem}

\hypersetup{
	colorlinks,%
	citecolor=black,%
	filecolor=black,%
	linkcolor=black,%
	urlcolor=black
}

\pgfplotsset{
	grid style = {
		dash pattern = on 0.025mm off 0.95mm on 0.025mm off 0mm, 
		line cap = round,
		black,
		line width = 0.5pt
	},
	tick label style={font=\small},
	label style={font=\small},
	legend style={font=\footnotesize},
}

\pgfplotscreateplotcyclelist{laneas1}{
	cyan!60!black,		solid, 	every mark/.append style={fill=cyan!60!black},mark=x\\%
	cyan!60!black,		solid, 	every mark/.append style={fill=cyan!60!black},mark=+\\%
	cyan!60!black,		solid, 	every mark/.append style={fill=cyan!60!black},mark=o\\%
	red!80!black,       dashed\\%
	cyan!60!black, 		densely dotted,	every mark/.append style={fill=cyan!80!black},mark=diamond*\\%
}

\ifCLASSINFOpdf
\else
\fi
\newtheorem{Lemma}{Lemma}

\begin{document}
	%
	\title{A Spatial  Basis Coverage Approach For Uplink Training And Scheduling In Massive MIMO Systems \thanks{This research has been supported by  the European project One5G.}}
	
	\author{\IEEEauthorblockN{$\text{Salah Eddine Hajri}^* \text{and Mohamad Assaad}^*$}\\
		\IEEEauthorblockA{*TCL chair on 5G, Laboratoire des Signaux et Systemes (L2S, CNRS), CentraleSupelec,	91190	Gif-sur-Yvette, France\\
			Email: \{Salaheddine.hajri,\; Mohamad.Assaad\}@centralesupelec.fr}
		
	}
	
	%

	
	\maketitle
	
	\begin{abstract}
Massive multiple-input multiple-output (massive MIMO) can provide large spectral and energy efficiency gains. Nevertheless, its potential is conditioned on acquiring accurate channel  state information (CSI). In time division duplexing (TDD) systems, CSI is obtained through uplink training which is hindered by pilot contamination. The impact of this phenomenon can be relieved using spatial division multiplexing, which refers to partitioning users based on their spatial information and processing their signals accordingly. The performance of such schemes depend primarily on the implemented grouping method. In this paper, we propose a novel spatial  grouping scheme that aims at managing  pilot contamination while reducing the required training overhead in  TDD massive MIMO. Herein,  user specific decoding matrices are derived based on the  columns of the discrete Fourier transform matrix (DFT),  taken as a spatial basis. Copilot user groups are then formed in order to obtain the best coverage of  the spatial basis with minimum overlapping between  decoding matrices. We provide two algorithms that achieve the desired grouping and derive their  respective performance  guarantees. We also address  inter-cell  copilot interference through efficient  pilot sequence allocation,  leveraging the  formed copilot  groups.  Various numerical results are provided to showcase the efficiency  of the proposed  algorithms.
	\end{abstract}
	

	%
	\IEEEpeerreviewmaketitle
	\pagebreak
	
	\section{Introduction}

	Wireless networks are under the strain of  exponentially increasing  demand for higher data  rates. In order to  provide the  required performance  leap,   several new physical  layer technologies can be relied upon. One of the most promising technologies is, without doubt,   massive MIMO \cite{Noncooperative}. By  leveraging a large number of  antennas at the BSs,  massive MIMO  is able to  provide considerable improvement  in the   network's spectral and energy  efficiencies with simplified   transceiver design   \cite{mMIMO_antennas},\cite{EE_marzetta}. This promoted massive MIMO to   be  a key enabler of  future generation networks \cite{10Myth}. 
	
		However, massive MIMO gains depend heavily on acquiring accurate  CSI estimates at the BSs. In TDD systems, CSI estimation is performed through uplink (UL) training, leveraging channel reciprocity \cite{pilot_reduction1}.  Owing to the limited coherence interval, the training resources  are restricted and the same  pilot sequences  need to  be  reused resulting in pilot contamination \cite{pilot_reduction1}.  Addressing this issue has lead to the development of numerous CSI estimation methods  that exploit different channel statistics in order to mitigate copilot interference and enhance CSI accuracy.  We refer to the survey \cite{pilot_contamination_survey} for more information. Several of these methods leverage spatial division multiplexing (SDM).
	 
				In \cite{coordinated}, the authors proposed a coordinated approach to channel estimation. They proved that exploiting covariance information, under certain conditions, can lead to a complete removal of pilot contamination  when the number of antennas grow very large. The authors showed that the channel estimation performance is a function of the degree to which the eigenspaces of the desired and interference signals overlap with each other.	In \cite{caire_oppor}, Joint spatial division and multiplexing (JSDM) for	multi-user MIMO downlink (DL) was investigated. JSDM is a scheme that aims to serve users by clustering them such that users within a group have approximately similar channel covariances, while users across groups have near orthogonal covariance eigenspaces.  JSDM  has been  designed, originally, for FDD MIMO systems without considering  inter-cell interference.  In \cite{metis}, an adaptation  of JSDM to  the TDD case was proposed, taking into consideration  inter-cell interference.  In \cite{metis}, density-based  spatial  clustering  of  applications with  noise  (DBSCAN)  algorithm was used instead of   $K$-mean in order to  avoid the otherwise additional requirement of estimating the number of eigenspace-based clusters $K$. In  \cite{user_clustering_conf}, user grouping based on  channel direction was proposed. The authors proposed to  schedule  copilot  user such that  their channels are semi-orthogonal. This  results in enhancing CSI accuracy  which, consequently, improve  the  achievable spectral  efficiency (SE). In \cite{unified}, a unified scheme  for TDD/FDD massive MIMO systems was proposed. Based on a spatial basis expansion model (SBEM), CSI estimation overhead was reduced and pilot contamination relieved. Each user is associated with a spatial  signature that groups the index set of nonzero DFT points of its channel. Users are then  grouped such that Users in the same group have non-overlapping spatial signatures.  Each group is then allocated a pilot sequence and  signals are discriminated based on the different spatial signatures.

	    Previous works on  spatial multiplexing proposed to  group users based on their  channel covariance  eigenspace \cite{tot,caire_oppor}, when the  channel covariances are known, or simply based on their signals mean direction of arrivals \cite{unified}. The performance of  such  methods is determined by the user grouping scheme. In \cite{caire_oppor}, a chordal distance based $K$-mean clustering was proposed for FDD systems with  JSDM. In \cite{tot}, the authors investigated a wide range of  similarity measures such  as 	weighted likelihood, subspace	projection and Fubini-Study based similarity measures. The authors introduced two clustering methods namely, hierarchical and $K$-medoids clustering for user grouping with the aforementioned proximity measures. A comparison of the proposed grouping methods was performed and the combination that achieves the largest capacity was derived. In \cite{unified}, a greedy user scheduling algorithm was implemented in order to partition users into copilot groups based on their spatial signatures. Although a considerable increase in SE was recorded, these methods come with a number of shortcomings. In fact, in order to achieve a good user clustering,  the $K$-medoids and $K$-mean clustering methods require a prior  estimation of the  parameter $K$.  In addition, these methods use an averaging in order to derive the group specific eigenspace matrix which can lead to a  substantial overlapping between  clusters. Practically, users might have similar but not necessarily identical second order  channel statistics. This dictates  the need to  consider individual  user spatial information. In addition, the efficiency of leveraging the totality of the available 	degrees of freedom (DoF) should be more emphasized. 
%

				In this work, we propose a novel spatial user grouping scheme.  We consider a multi-cell TDD massive MIMO system, in which, spatial diversity is exploited in order to allow for a more pilot reuse, within each cell, while mitigating copilot interference.  This allows to increase in the number of scheduled users for the same training overhead while improving SE. Since a  multi-cell  system is considered, both intra and inter-cell  pilot contamination are tackled. We  choose to  decouple these two  problems  and  address them successively.
				In order to  deal with  intra-cell interference,  we propose a spatial grouping and scheduling scheme. We construct copilot groups based on the users spatial signatures. In each cell, any given  copilot group is formed such that it contains users with minimum overlapping in their signals spatial signatures and that provide a maximum  coverage of the signal spatial basis. The proposed approach is referred to  as  \textit{Spatial basis coverage  copilot user selection.} The idea is to  associate each user with a set of beams that concentrate a large amount of its channel power. Since uniform linear arrays (ULAs) are considered, the columns of  a unitary DFT matrix are used as spatial basis \cite{caire_oppor},\cite{unified}. After obtaining the  users specific decoding matrices, the BSs derive  copilot groups.  Each group  provides a maximum coverage of  all available independent DFT beams with  minimum  overlapping between users specific beam  matrices.   This approach enables also to couple the problems of user grouping and scheduling which reduces the complexity  of the network management. We provide two formulation of the copilot grouping problem and we propose two grouping  algorithms accordingly. First, the problem  of  copilot user selection is  formulated as a \textit{maximum coverage problem} \cite{MCP}. This  formulation  enables to  derive a low complexity   algorithm that  provides at least an $(1-(\frac{\tau-1}{\tau}) ^{\tau} )(1-\frac{1}{e})$-approximation of the optimal grouping. In the second case,  the problem  of copilot  group  generation is formulated as a \textit{Generalized maximum coverage problem} \cite{GMC}.   We propose a  low complexity  algorithm that  provides at least an $ (1-(\frac{\tau-1}{\tau}) ^{\tau} )\frac{\frac{3}{2}- \frac{e^{-2}}{2}}{1- e^{-2}}$-approximation  of the optimal user grouping.
				
				 	Based on the  constructed copilot groups, we address  inter-cell copilot interference through an  efficient cross-cell pilot  allocation. We propose a graphical framework based on the copilot groups spatial signature.  Using this information,	the network is able to allocate specific UL training sequences to copilot  groups, such that the previously defined spatial receivers can manage cross-cell copilot interference. The resulting pilot allocation problem is formulated as a \textit{max-$\tau$-cut  problem} \cite{maxk}, which enables to use a low complexity algorithm that provides a $(1- \frac{1}{\tau})$- approximation of the optimal solution.\\

	\textit{\textbf{Notations:} } We use  boldface small letters $(\textbf{a})$ for vectors,  boldface capital letters $(\textbf{A})$ for matrices. The notations $\textbf{A}^{\dagger}$ and $\Tr(\textbf{A} ) $  are used for Hermitian transpose and trace of matrix $\textbf{A}$, respectively .  $\left\| \textbf{a} \right\| $ denotes the euclidean norm of vector $\textbf{a}$. $\textbf{I}_n$ is used to denote the $n \times n$ identity  matrix and $\lVert A\rVert_F  ^2 = \Tr(\textbf{A}^{\dagger} \textbf{A}) $ denotes the  square of the Frobenius norm of $\textbf{A} $.

	\section{System Model And Preliminaries}\label{spatial_multiplexing:sec:sysmodel}	
	
	We consider a multi-cell, multi-user massive MIMO network  operating in TDD mode. The system is  composed of  $N_c$ cells containing, each,  a {BS} that is equipped with a large $M$-element ULA. 
	Each  BS is serving $K$ single omni-directional antenna users such that $K>>M$. Users are randomly  distributed in each cell. 	As introduced above, in this paper we focus on a TDD massive MIMO system, where the entire frequency band is used for DL and UL transmission by all BSs and users.  Since a TDD system is  considered,  the focus will be on addressing the UL training bottleneck. Considering flat fading channels, the channel vector between  user $i$ in cell $b$ and the BS of the  $r^{th}$ cell, $\textbf{g}^{[r]}_{ ib }$ is composed of an arbitrary number of i.i.d.  $P$ rays ($P>>1$)\cite{coordinated}. Hence, the  UL channel $\textbf{g}^{[r]}_{ ib }$ is given by the  following multi-path model
	
	\begin{align}\label{eq:channel_decomp_spatial}
	& \textbf{g}^{[r]}_{ ib } =   \frac{1}{\sqrt{P}} \sum\limits_{p=1}^{P} \textbf{a}(\theta^{[r,p]}_{ib}) \gamma^{[r,p]}_{ib}, 
	\end{align}

	Here, $\gamma^{[r,p]}_{ib}$ represents the  complex gain of  the  $p^{th}$ ray from user $i$ in cell $b$ and the BS of the  $r^{th}$ cell and follows  a $ \mathcal{C}\mathcal{N} \left( 0, \mu^{[r]^2}_{ib} \right)$ distribution where $\mu^{[r]}_{ib} $ denotes the  average  attenuation of  the channel.  $\theta^{[r,p]}_{ib}$ denotes the  direction of arrival (DOA)  of the   $p^{th}$ ray from user $i$ in cell $b$ and the  BS of the  $r^{th}$ cell. Moreover,  $ \textbf{a}(\theta^{[r,p]}_{ib}) \in \mathbb{C}^{M\times1 }$ is the array  manifold vector which is given  by:
	
	\begin{align}\label{eq:manifold}
	& \textbf{a}(\theta^{[r,p]}_{ib})= \big[1, e^{\frac{j 2\pi d }{\lambda}sin(\theta^{[r,p]}_{ib})},\ldots, e^{\frac{j 2\pi d }{\lambda}sin(\theta^{[r,p]}_{ib})(M-1)} \big], 
	\end{align}
	where $\lambda$ denotes the  signal wavelength,  $d$ refers to  the antenna spacing such that $d \leq \frac{\lambda}{2}$.  As in  \cite{coordinated} and  \cite{unified}, the  incident  angles of  each user, with mean DOA $\theta^{[r]}_{ib}$, are  considered to  be restrained in  a narrow  angular range $\big[\theta^{[r]}_{ib}-\Delta\theta^{[r]}_{ib} ,\theta^{[r]}_{ib}+\Delta\theta^{[r]}_{ib}   \big]$. Within this  range    $\textbf{a}(\theta^{[r,p]}_{ib}),p=1,\ldots,P, \forall i,b,r $ are mutually correlated. Consequently, the covariance matrix of each  channel $\textbf{g}^{[r]}_{ ib }$, which is given by  $\textbf{R}^{[r]}_{ ib }= \E{\textbf{g}^{[r]}_{ ib } \textbf{g}^{[r]^\dagger}_{ ib } }$,  possesses a  low-rank property.

	The BSs acquire CSI  estimates  using orthonormal training sequences (i.e., pilot sequences) in the UL.  For that,  we consider a set of orthonormal training sequences, that is, sequences $q_i\in \mathbb C^{\tau \times 1}$ such that $ q^{\dagger}_i q_j =\delta_{ij} $ (with $\delta_{ij}$ the Kronecker delta). In this paper, we consider an aggressive pilot reuse approach. In fact, in  addition to  reusing the  same set of orthogonal pilot sequences in every cell, we consider that the same sequences are reused even within  one  cell. Consequently,	 the channel estimates are corrupted by	both inter and intra-cell  pilot contamination.

	\section{Spatial division multiplexing based user scheduling}
	
	\subsection{Spatial basis in massive ULAs}

	Massive MIMO systems  provide a substantial  SE  gain by spatially multiplexing a large number of  mobile devices. This greater number of served devices requires higher signaling or feedback overhead in order to  obtain  CSI estimates.  This issue  promoted many  research work which  resulted in designing  new transmission  strategies for massive MIMO  leveraging low-rank  approximation  of  the  channel  covariance matrix \cite{coordinated,tot,caire_oppor}.  Indeed, based on the fact that  the  incident  signals  at the BSs are characterized by narrow angular spread,  it  was proven  that  the  effective channel  dimension can be reduced with negligible  capacity loss.
			Exploiting the channel low-rank property  in  massive MIMO transmission strategies proved to provide non-negligible gains in performance for both TDD and FDD systems \cite{unified,caire_oppor,fdd_conf}. In FDD mode, spatial division multiplexing allows to reduce the CSI feedback overhead while incurring no capacity loss  \cite{caire_oppor}.  In TDD, such methods enable to reduce training overhead while mitigating the impact of pilot contamination \cite{unified}.   These gains are mainly due to the capacity of spatial division methods to  utilizes the independent spatial  spaces of different users in order to discriminate between their signals. Spatial division based methods rely on  efficient spatial  information-based user grouping.  Several works proposed to  perform  this grouping using the  $K$-mean  algorithm  with different proximity  measures \cite{caire_oppor,tot}. $K$-medoids and hierarchical clustering have also been proposed \cite{tot}. A DFT-based greedy user grouping was also considered in  SBEM  \cite{unified}.

	In this  paper,  we provide  an  alternative  low-rank CSI estimation approach, leveraging  the characteristics of  ULAs. Indeed,  in this case, it was proven that a unitary DFT matrix constitute a good spatial basis  for  the signal \cite{unified , caire_large}. This means that, without accurate estimates of the  channel   covariance matrices,   spatial  division multiplexing  can be implemented using a unitary DFT matrix. For each user,  it  is  sufficient  to  derive decoding matrices based on  the   DFT matrix vectors that span the majority of its channel power \cite{unified}. The proposed spatial grouping approach  differs from previously propose ones,  \cite{caire_oppor,tot,unified},  as it   takes into consideration the  efficiency  of  the  spatial  space  coverage. The  main idea is to  group users according to their spatial  signatures  and  allocate pilot sequence  such that  copilot users are spatially separated (i.e., their spatial  signatures span independent subspaces). Copilot user grouping is  performed  based on two criterion. First, users are grouped such that  their spatial  signatures' overlap is minimal. Second, the  users in each group  provide  maximum  coverage  of all  DFT vectors. It is  worth  mentioning that  the  latter criterion was arguably under-leveraged in previous work.

	\subsection{Obtaining spatial information}
	
	Low-rank CSI estimation methods rely on different  knowledge levels of the channel statistics. In our case,  the proposed method requires the knowledge of each user spatial  signature (i.e. the main DoAs of each user signal).  
   This information needs to be estimated. Owing to the slow varying nature of spatial information,  an efficient estimation can be achieved with low signaling overhead. Indeed, the spatial signatures can be obtained using an UL preamble or DL training followed by feedback. We shed some light on how spatial signatures are obtained.

     UL preamble\cite{unified} can be used in order to  obtain the users spatial signatures. In this case, we need to  include an additional UL spatial training period. This should require a minimum  of $K$ training resource elements. During the UL preamble, the BSs receive the spatial training signal and derives the DFT-based decoding matrix for each user. In this work we choose to include a design parameter   $0<\alpha<1$ in order to quantify the importance of each DoA for each channel. 
	Particularly, for each  user   $i,b$   the   DFT-based decoding matrix $\textbf{F}^{[b]}_{ib}$ is obtained as follows
\begin{align}\label{eq:projection_matrix}
&   \textbf{F}^{[b]}_{ib}= \{ \textbf{f}_{s} \in \textbf{F},     \frac{\left\|  \textbf{g}^{[b]^\dagger}_{ ib} \textbf{f}_{s}\right\|^2 }{{Tr(\textbf{R}^{[b]}_{ib})}} \geq \alpha \},
\end{align}	

	$\textbf{F}^{[b]}_{ib}$ will be used as the bases in which the  user's signal is detected and will henceforth be refereed to as spatial  signature of  user $ib$.
	
	The spatial signature of  each user can be also  inferred from DL  spatial training followed by feedback \cite{user_clustering_conf}. In this case, users quantify the power of their channel along each direction and  feedback the indexes of the chosen 	DFT vectors according to~\eqref{eq:projection_matrix}. 
	The  spatial  signature of  each user forms a subspace that concentrate a large percentage  of  its channel power. Consequently,  allocating the  same  copilot  sequence  to users with  minimum overlapping in their spatial  signature (\ref{eq:projection_matrix}) enables to  discriminate between their  signals since  the power of their channels is  concentrated in  different  subspaces. 		
	In order to gain more insight  on the needed criterion to achieve efficient  spatial  grouping,  we analyze the  power of  the desired and interference signals, when the considered  DFT-based decoding matrices are used.

	\subsection{Achievable Performance With Adaptive Spatial Division Based User Scheduling}

For every user $i,b$, we consider the matrix $\textbf{F}_{ib}$  formed by the  vectors of the  DFT matrix according to~\eqref{eq:projection_matrix}. During the UL training period, active users send their pilot sequences. For analytical simplicity, we consider that  UL training sequences have the same reuse factor within all cells. We consider that the  network schedules $N_{p}$ users to use any given pilot sequences in each cell. During UL training, the  received pilot signal $\textbf{Y}_{p}^{[b]}$ at BS  $b$ is given by:
	\begin{align}\label{eq:uplink_training}
	& \textbf{Y}_{p}^{[b]}  =  \sqrt{\rho_p}\sum_{r=1}^{N_c} \sum_{l=1}^{\tau}   \sum_{i \in \Sigma (l,r)}^{}     \textbf{ g}^{[b]}_{ir}  \; {\textbf{q}_{l}}^\dagger  + \textbf{W}_p,
	\end{align}
	where $\textbf{W}_p \in \mathbb{C}^{ M \times \tau} $ refers to  an additive  white  Gaussian  noise matrix with i.i.d. $\mathcal{C}\mathcal{N}(0,1)$ entries and $\rho_p$ denote the  pilot transmit power.   $\Sigma (l,r)$ denotes the  set of user in cell $r$ that are using pilot sequence $l$ during UL training.
	The  $b^{th}$ BS  then uses the  orthogonality of  the training sequences in order to  obtain  the  Least square (LS) estimate of the  channel of user $i,b$.  The    BS  estimates the  channel  of each user $i,b$ after projecting the received signal on   $\textbf{F}^{[b]}_{ib}$ as
	\begin{align}\label{eq:channel_estimation}
	& \hat{\textbf{g}}^{[b,ib]}_{ib}= \textbf{F}^{[b]^\dagger}_{ib}   ( \frac{\textbf{Y}_{p}^{[b]}{\textbf{q}_{\chi(i,b)}}}{\sqrt{\rho_p}}),
	\end{align}
	where $\chi(i,b)$ denote the index of  the  training  sequence used by user $i,b$. Note that  indexes of the  projection matrix $\textbf{F}^{[b]}_{ib}$ are added to the channel  estimate since its distribution depends on user $i,b$ spatial  signature.
 During UL data transmission, BS $b$ receives the  following  data signal
	\begin{align}\label{eq:uplink_data}
	& \textbf{Y}_{u}^{[b]}  =  \sqrt{\rho_u}\sum_{r=1}^{N_c} \sum_{l=1}^{\tau}   \sum_{i \in \Sigma (l,r)}^{}     \textbf{ g}^{[b]}_{ir}  d_{ir}  + \textbf{w}_u,
	\end{align}
	where $\textbf{{w}}_u \in \mathbb{C}^{ M \times 1}$ refers to  an additive  white  Gaussian  noise vector with i.i.d. $\mathcal{C}\mathcal{N}(0,1)$ entries and $\rho_u$ denotes the UL data transmission power. 
	We consider linear detection where, the  signal  of  each user $i,b$ is estimated using a precoded maximum ratio combining receiver. In order to detect the signal  of user $i,b$, BS $b$ uses $  \frac{ \textbf{F}^{[b]}_{ib} \hat{\textbf{g}}^{[b,ib]}_{ib} }{\lVert\hat{\textbf{g}}^{[b,ib]}_{ib}\rVert} $ as receive filter. The estimate of the signal of user $i,b$ can be decomposed as follows:
	\begin{align}\label{eq:uplink_detection}
	 \hat{d_{ib}}  =&
	  \underbrace{{\hat{\textbf{g}}^{[b,ib] \dagger}_{ib} }  \hat{\textbf{g}}^{[b,ib]}_{ib}  d_{ib}}_{DS^{[b]}_{ib} } 	+\underbrace{\sum_{r=1}^{N_c}  \sum_{ \substack{ u \in \Sigma (\chi(i,b),r) \\ ur \neq ib}}^{}     {\hat{\textbf{g}}^{[b,ib] \dagger}_{ib} }  \hat{\textbf{ g}}^{[b,ib]}_{ur}  \; d_{ur}}_{I^{[b,C]}_{ib} }+    \underbrace{{\hat{\textbf{g}}^{[b,ib] \dagger}_{ib} }   \frac{\textbf{w}^{[ib]}_u}{\sqrt{\rho_u}}}_{N^{[b]}_{ib} } \\
	  		&+ \underbrace{\sum_{r=1}^{N_c} \sum\limits_{ \substack{ u \in \Sigma (\chi(i,b),r) }}^{}    {\hat{\textbf{g}}^{[b,ib] \dagger}_{ib}  }  \tilde{\textbf{ g}}^{[b,ib]}_{ur}  + 
	  			\sum_{r=1}^{N_c} \sum_{l \neq \chi(i,b) }^{\tau}   \sum_{u \in \Sigma (l,r)}^{}   {\hat{\textbf{g}}^{[b,ib] \dagger}_{ib} }   \textbf{ g}^{[b,ib]}_{ur}  \; d_{ur}}_{I^{[b,NC]}_{ib}} 
	  		, \nonumber
	  	\end{align}

	where $\tilde{\textbf{g}}^{[b,ib]}_{ib} $ represents the channel  estimation error $\forall i,b$ which  is  defined as ${\textbf{g}}^{[b,ib]}_{ib}= \hat{\textbf{g}}^{[b,ib]}_{ib}+\tilde{\textbf{g}}^{[b,ib]}_{ib}$. Here  $DS^{[b]}_{ib}$,  $I^{[b,C]}_{ib}$,   $N^{[b]}_{ib}$ and  $I^{[b,NC]}_{ib}$ refer to  the desired signal,  the impact of copilot  user interference,  additive noise and   non-coherent interference, respectively. 
	As we can see in ~\eqref{eq:uplink_detection}, the power of  copilot interference depends on to which degree the spatial  signatures of copilot users overlap.  Although increasing the pilot reuse factor increases the number of active users, intra-cell copilot interference can be completely mitigated if the same pilot sequence is allocated to users with non-overlapping spatial  signatures. Scheduling an excess of users with  non-overlapping spatial  signatures implies that all  the  system  DoFs need to be exploited when  possible.
	Consequently, selecting copilot users, in each cell, is then of paramount importance.  In the next section,  we  address the problem of intra-cell  copilot interference through appropriate spatial signature-based user grouping.

	\section{An Alternative Approach  to Spatial User Grouping: A Spatial Basis Coverage Problem}

	Although  previously proposed grouping approaches get the work done and  provide considerable performance increase for both FDD and TDD modes, they suffer, nevertheless, from a range of shortcomings that may limit the potential of spatial division multiplexing. Indeed, maximum coverage of all the available  DoFs should be emphasized. When applied to  the spatial division problem, classical clustering approaches concentrate on the mutual distance between user channels subspaces with little regard to the final coverage of independent streams. This means that, although the condition of independent spatial information is met, the DoFs  that the massive MIMO  system provides can be underexploited. Addressing these shortcomings can help boost the performance of spatial division multiplexing methods.

	In this paper, we propose a method that, in addition to the requirement of independent spatial subspaces, emphasizes on leveraging  all  DoFs of massive MIMO system.  The same pilot sequence is allocated to users with minimum spatial  signature overlapping which constitute a copilot group within each cell. In addition, the users within each copilot group achieve a maximum  coverage of  all interdependent DoAs. Consequently, in each cell,  a total  of $\tau $ copilot  groups need to  be constructed, each of which is associated with a distinct  training sequence.  In what follows, we formulate the spatial basis coverage copilot user selection problem. We then  provide  efficient algorithms  that  enables to  solve it.  Two  formulations are considered. 
	 
		The first approach is  power agnostic. This  means  that it  neglects the channel  power and  concentrates only on minimum spatial  signature overlapping and maximum coverage  without  discriminating users. 
	The  second approach is  power aware. This means that  users are  prioritized based on their  achievable channel gains in addition to the  criterion considered in the  power agnostic approach.
	The  differences  between  the  two approaches lie mainly  in  complexity, achievable SE gains and fairness. These differences  will be  discussed  in more details  further  in this paper.

	\begin{figure}[!htb]
		\centering	
		\includegraphics[scale=.6]{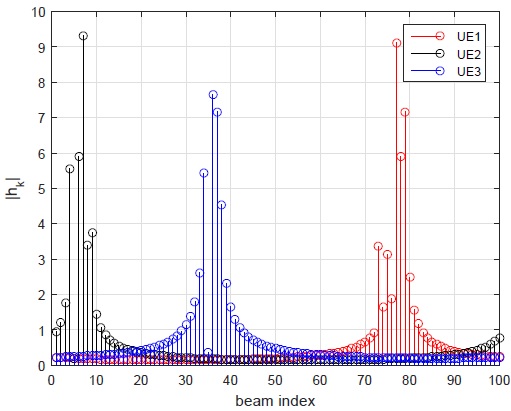}
		\captionsetup{font=footnotesize}
		\caption{Example of  spatial basis coverage for  $M=100$ \label{cover_fig}}
	\end{figure}
	\subsection{Power Agnostic Spatial Basis Coverage}

	In this subsection, we focus on solving the spatial basis coverage  problem in the power agnostic approach. In this case, the BSs know only the set of DFT beams that concentrate a large percentage of channel power (i.e, $\textbf{F}^{[b]}_{ib}, \forall i,b$). The user-beam association is performed as specified in the previous section thanks to DL training~\eqref{eq:projection_matrix} \cite{user_clustering_conf} or  using  UL  preamble \cite{unified}.
	As previously discussed, the BSs perform spatial basis coverage based copilot user selection in order to  schedule users for UL  training. This is  done in order to deal with intra-cell  copilot interference. The out of cell copilot  interference will be addressed later in this paper. In the power agnostic case, the BSs  do not take into  consideration  the achievable gain along  each beam. Consequently,  in this case, the problem  reduces to  scheduling users with minimum  spatial signature overlapping and  maximum  coverage of  the DFT beams.  This actually  simplifies  the  problem at  hand and  enables to derive  the  desired grouping  with low complexity. The  power agnostic  approach is  also characterized by  the  upside  of fairness  since  it does not  discriminate  scheduled  users based on their  channel  gain.  
	However, this  means that more flexibility should be allowed when constructing copilot users. In  fact, since we cannot prioritize users based on their channel gains, it  may be wise to allow for some spatial   overlapping. We also allow for another  degree of flexibility  in this problem, namely, pilot reuse in  each cell. In fact, in this  work,  we consider that the reuse factor  of  each pilot sequence  can vary from one cell  to the other. This implemented in order to allow for a more flexible specific training  sequences allocation to the  copilot  groups when dealing with  inter-cell copilot interference.
	The main principle of the  spatial basis coverage problem is depicted in Figure (\ref{cover_fig}).

	We consider $\tau$ copilot groups (covers) per cell $C^{[b]}_k,k=1,\ldots,\tau, b=1,\ldots,N_c$. Each  copilot group in each cell will  be associated with  a distinct pilot sequence. 
	We start by defining $x^{[k]}_{i,b}, \;\forall i,b,k $ and $y^{[k]}_{s,b} , \;\forall s=1,...,M, b=1,...,N_c, k=1,...,\tau$, which  are given by	
	\begin{equation}\label{eq:x_comb}
	\begin{aligned}
	& x^{[k]}_{\{i,b\}}   =
	\left\{
	\begin{array}{ll}
	1  & \mbox{if  user $i,b$ is selected in copilot group $C^{[b]}_k$.}  \\
	0 & \mbox{otherwise. } 
	\end{array}
	\right.\\
	& y^{[k]}_{s,b}   =
	\left\{
	\begin{array}{ll}
	1  & \mbox{if  beam $\textbf{f}_s$ is covered in cell $b$ and  copilot group $C^{[b]}_k$.}  \\
	0 & \mbox{otherwise. } 
	\end{array}
	\right.\\
	\end{aligned}
	\end{equation}
	
	Formally, under the  power agnostic approach, the spatial basis coverage based copilot UE selection problem  can be formulated as follows:
	
	\begin{align}\label{eq:problem_coarse}
	\underset{Y}{\text{max}} & \sum_{k=1}^{\tau}  \sum_{b=1}^{C} \sum_{s=1}^{M}  y^{[k]}_{s,b}\\\tag{10a}
	\text{subject to} \;  & \sum_{i}^{}  x^{[k]}_{\{i,b\}}  \leq U^{[k]}_{b} \;\; \forall k=1...\tau, \;\;\forall b=1,...,N_c \\\tag{10b}
	& \sum_{i, f_s \in F_{ib} }^{} x^{[k]}_{\{i,b\}}   \geq y^{[k]}_{s,b}  \;\; \forall k=1...\tau,\;\;\forall b=1,...,N_c, \nonumber
	\end{align}
	$(10a)$   guarantees that the number of users in a given copilot group $C^{[b]}_k,k=1,\ldots,\tau, b=1,\ldots,N_c$,  is upper bounded by $U^{[k]}_{b},k=1,\ldots,\tau, b=1,\ldots,N_c$. Note that $U^{[k]}_{b}$ is a design parameter that defines the  reuse factor of a given pilot sequence in each cell. 
	Depending on the considered setting, $U^{[k]}_{b}$ can be the same or differs from one cell to the other. 
	$(10b)$    guarantees that,  for any covered  beam $f_s$  in cell $b$, in   copilot group $C^{[b]}_k$, at least one user $i,b$ with $\textbf{f}_s \in \textbf{F}^{[b]}_{ib}$ is scheduled for UL training in copilot group $C^{[b]}_k$. 
	We start by showing the computational  intractability of problem~\eqref{eq:problem_coarse}.
	
	\begin{Lemma}\label{spatial_multiplexing:theor:NP_coarse}
		The  spatial basis coverage based copilot UE selection problem~\eqref{eq:problem_coarse} is  NP-hard. 
	\end{Lemma}
	\begin{proof}
		For $C=1$ and $\tau=1$,~\eqref{eq:problem_coarse} is equivalent to a \textit{maximum coverage problem} which  is known to be NP-hard \cite{GMC}. Consequently,~\eqref{eq:problem_coarse} is also  NP-hard. 
	\end{proof}
 In order to  obtain an efficient suboptimal solution of~\eqref{eq:problem_coarse}, we use two nested greedy phases. In the upper phase, the algorithm produces $\tau$ maximum  coverages of the DFT matrix vectors ($\textbf{F}$), in each cell. The maximum covers $C^{[b]}_{k}, k=1..\tau, b=1,..,N_c$, are computed successively in  a greedy manner.  Each of the  maximum  covers  is  computed using another greedy method that  goes as follows. For each $C^{[b]}_{k}, k=1..\tau, b=1,..,N_c$, the set of  uncovered beams  is initialized as the  vectors of the DFT matrix. Then  users are added to $C^{[b]}_{k}$ successively while  selecting, at each iteration, the  user  with the spatial  signature  that  cover a maximum of the uncovered  DFT columns. This  procedure  is  repeated until attaining  the  reuse constraint $U^{[k]}_{b}$  for  each  copilot group,  in each  cell.
	Different from previously proposed algorithms, the present approach enables to satisfy the spatial independence requirements within each copilot group  while offering a maximum utilization of the excess of DoFs.
	The detailed algorithm  is given in table \ref{agnostic}. We denote by  $\Gamma(b)$ the set of users in cell $b=1,\ldots,N_c$.
	\begin{center}
		\begin{tabular}{ l  }
			\hline
			\hline
			\emph{Initialize}: Copilot groups sets $C^{[b]}_k= \emptyset,k=1,\ldots,\tau, b=1,\ldots,N_c$, \\

			User specific beam  matrices $\textbf{F}_{ib},  \forall i \in \Gamma(b) ,b=1,\ldots,N_c$\\
	
			$1. $\textbf{For} $b=1:N_c$ \textbf{do}:\\
	
			$2. $\textbf{For} $k=1: \tau$ \textbf{do}:\\
		
			$3. $Define the set $Un=\textbf{F}$ as the set of uncovered beams.\\
	
			$4. $\textbf{For} $j=1: U^{[k]}_{b}$ \textbf{do}:\\ 
	
			$5. $ $i^*  \longleftarrow     \underset{i \in \Gamma(b)}{\text{argmax }}  \lvert \textbf{F}_{ib} \cap  Un  \rvert $\\
		
			$6. $ $Un \longleftarrow Un \setminus \{\textbf{F}_{i^*b} \cap  Un\}$\\
		
			$7. $ $C^{[b]}_{k} \longleftarrow C^{[b]}_{k} \bigcup   i^*$\\
		
			$8. $ \textbf{End for}\\
	
			$9. $ \textbf{End for}\\
		
			$10.$ \textbf{End for}\\
			\hline
			\hline
		\end{tabular}
		\captionof{table}{\textbf{Power Agnostic Spatial Basis coverage based copilot UE selection}\label{agnostic}}
	\end{center}
	The algorithm in  table~\ref{agnostic} produces $\tau$ copilot user groups in each cell. Each  copilot group  maximizes the  coverage  of  the DFT vectors while minimizing the  overlapping between copilot users spatial signatures. Consequently, the CSI estimation of each user can be enhanced by a simple linear projection. Note that the proposed algorithm allows for some subspace overlapping. This  is actually needed since users are not prioritized based on their  channel gains. We now proceed by deriving the performance guarantee of the proposed algorithm.
	\begin{Theorem}\label{spatial_multiplexing:theor:ratio_coarse}
		The  algorithm in table~\ref{agnostic}  provides an $(1-(\frac{\tau-1}{\tau}) ^{\tau} )(1-\frac{1}{e})$-approximation of the optimal solution of  problem~\eqref{eq:problem_coarse}.
	\end{Theorem}
	\begin{proof}
		See \ref{ratio_coarse}.
	\end{proof}

	\subsection{Power Aware Spatial Basis Coverage}

	In the  power aware approach, users are prioritized  based on the power of their signals along each direction.   The  resulting problem provides a more efficient  grouping since it  takes  into  consideration the overlapping between copilot users spatial signatures,  the  coverage of the signal  space and  the  power of each user  channel. Such grouping provides a higher SE gain.  However, this improvement comes at the price  of  augmented complexity since  users are discriminated based on their  channel gains.
	In this case, we define a different  value for each  beam depending on which user is covering it. For each user $i,b$, the  value  associated with beam $\textbf{f}_s, s=1,\ldots,M$ is given by  $\zeta^{[s]}_{ib}$, where $\zeta^{[s]}_{ib}$ is the  power of  user $i,b$ channel  along $\textbf{f}_s, s=1,\ldots,M$. This consideration changes the formulation of the spatial basis coverage based copilot UE selection problem~\eqref{eq:problem_coarse}.  The main idea of providing maximum  coverage of the DFT beams, in each cell and for each pilot sequence, still holds but the actual gain  associated with each beam  will also be taken into  consideration. The resulting combinatorial  optimization  problem can be formulated as follows
	\begin{align}\label{eq:problem_knowledge}
	\underset{Y}{\text{max}} & \sum_{k=1}^{\tau} \sum_{b=1}^{N_c} \sum_{i\in \Gamma(b)}^{}  \sum_{f_s\in \textbf{F}}^{}	\zeta^{[s]}_{ib} y^{[s,k]}_{\{i,b\}}     \\\tag{11a}
	\text{subject to} \;  & \sum_{i\in \Gamma(b), f_s \in F_{ib}}^{} y^{[s,k]}_{\{i,b\}}   \leq 1 \;\; \forall k=1...\tau, \;\;\forall b=1...N_c \\\tag{11b}
	& \sum_{i\in \Gamma(b), f_s \in F_{ib} }^{} x^{[k]}_{\{i,b\}}    \geq     y^{[s,k]}_{\{i,b\}}   \;\; \forall k=1...\tau,\;\;\forall b=1...N_c, \\\tag{11c}
	& \sum_{i\in \Gamma(b)}^{} x^{[k]}_{\{i,b\}}   \leq U^{[k]}_{b}   \;\; \forall k=1...\tau,\;\;\forall b=1...N_c,\nonumber
	\end{align}
	The constraints, in $(11a)$, guarantees that each beam is covered by at most one user.  
	$(11b)$    guarantees that,  for any covered  beam $f_s$  in cell $b$, in   copilot group $C^{[b]}_k$, at least one user $i,b$ with $\textbf{f}_s \in \textbf{F}_{ib}$ is scheduled for UL training in copilot group $C^{[b]}_k$. 
	$(11c)$ guarantees that  the total number of  the users associated with a given pilot sequence in a given  cell  is bounded.
	The difference between~\eqref{eq:problem_coarse} and~\eqref{eq:problem_knowledge} is mainly  the fact that  the actual  gain along each DFT beam  is taken into  consideration. This means that, the BS can optimize its pilot allocation accordingly with the final aim of maximizing the  total weight of the  covered DFT beams.  This means  that, in addition  to reducing copilot interference thanks to the non-overlapping  spatial  signatures, the users are selected such that  the achievable gain  along all the  available  independent  beams is maximized. We start by showing the computational  intractability of problem~\eqref{eq:problem_knowledge}.
	\begin{Lemma}\label{spatial_multiplexing:theor:NP_knowledge}
		The spatial basis coverage based copilot UE selection problem~\eqref{eq:problem_knowledge} is  NP-hard. 
	\end{Lemma}
	\begin{proof}
		For $C=1$ and $\tau=1$, the optimization problem~\eqref{eq:problem_knowledge} is equivalent to a \textit{Generalized Maximum  Coverage Problem} (GMC) which is known to be NP-hard. Consequently,~\eqref{eq:problem_knowledge} is also NP-hard.\\
	\end{proof}
	The  proof of computational  intractability  provides us with  insight on how to tackle problem~(\ref{eq:problem_knowledge}) efficiently. Indeed, in order to obtain an efficient suboptimal solution for ~(\ref{eq:problem_knowledge}), we adopt a successive coverage approach as the previous algorithm. The difference here comes in the construction of each copilot group  where the actual power along each beam needs to be considered. To this end, $ \forall \; k=1..\tau, b=1,..,N_c$, a  GMC problem is solved. The proposed approach is based on a modification of the coverage algorithm in \cite{GMC}. While,  in \cite{GMC}, a GMC is solved leveraging  a greedy procedure and Dynamic programming for Knapsack problems, we propose a two nested greedy phases in addition to replacing  Dynamic programming with the greedy algorithm for Knapsack problems. 
	Replacing  Dynamic programming by the greedy algorithm for Knapsack problems results in  reducing  the complexity of practical implementation. 
	
	The present approach enables to satisfy the spatial independence requirements within each copilot group  while offering a maximum utilization of the excess of DoFs. Since it takes into consideration the actual power of users signals, the present approach enables also to discriminate between users based on their channel gain in each direction. This, ultimately, results in more efficient utilization of the system's DoFs by prioritizing users  with high signal power in each direction.\\
	Before providing the detailed algorithm that addresses~\eqref{eq:problem_knowledge}, some definitions are now in order.  We  define a  user allocation $A$ as  a triple $A=(\phi,\xi,h)$, where   $\phi$ represents the set of selected users, $\xi$  denotes  the  set of corresponding covered beams  and $h$  is an assignment  from $\xi$ to $\phi$ such that $\forall   f_s \in \xi,\; h(f_s)$ denotes the user covering beam $f_s$. 
	For a given allocation $A$, we define $V(A)= \sum_{f_s \in \xi }^{} \zeta^{[s]}_{h(f_s),b}  $ as the value of $A$ and $W(A)= \lvert  \phi  \rvert $ as its weight.
	We also define the residual  value of $([i,b],f_s)$ with respect to $A$  as follows
	\begin{equation}\label{eq:residual_def}
	\begin{aligned}
	& V_A([i,b],f_s)  =
	\left\{
	\begin{array}{ll}
	\zeta^{[s]}_{i,b}  & \mbox{if  $f_s$ is not covered by A.}  \\
	\zeta^{[s]}_{i,b} -\zeta^{[s]}_{h(f_s),b} & \mbox{otherwise. } 
	\end{array}
	\right.\\
	\end{aligned}
	\end{equation}
	
 	The detailed algorithm  is presented  in tables  ~\ref{aware} and ~\ref{greedy}.

	\begin{center}
		\begin{tabular}{ l  }
			\hline
			\hline
			$1. $\textbf{For} $b=1:N_c$ \textbf{do}:\\
	
			$2. $	\textbf{For} $k=1: \tau$ \textbf{do}:\\

			$3. $	$A_g \longleftarrow $ \textbf{Greedy}($S= \{i,i \in \Gamma(b)\},b,k$ )\\
		
			$4. $	Find a single user $i^*,b$ with the highest power along its covered DFT vectors.\\
	
			$5. $	 $V(A_s) \longleftarrow   \sum_{\textbf{f}_s,\textbf{f}_s \in \textbf{F}_{i^*b}}^{} \zeta^{[s]}_{i^*b}  $\\
		
			$6. $	\textbf{If}  $V(A_g) \geq V(A_s)$:\\

			$7. $	$C^{[b]}_{k}  \longleftarrow   (A_g)$\\
	
			$8. $	\textbf{Else} \\
		
			$9. $	$C^{[b]}_{k} \longleftarrow  (A_s)$\\
	
			$10. $  \textbf{ End for}\\

			$11. $   \textbf{End for}\\
			\hline
			\hline
		\end{tabular}
		\captionof{table}{\textbf{Power Aware Spatial Basis coverage based copilot UE selection}\label{aware}}
	\end{center}
	\begin{center}
		\begin{tabular}{ l  }
			\hline
			\hline
			$1. $ $j \longleftarrow  0  $, $A = \emptyset$\\
		
			$2. $ \textbf{While} new DFT vectors with  positive residual  value can be added to $A$ without\\ 
	
			\quad violating the  cardinality  constraint $U^{[k]}_{r} $ ($ W(A) \leq U^{[k]}_{r} $)\textbf{do}:\\
		
			$3. $ Use the greedy algorithm for Knapsack problems in order to select a user $i^*$ and a subset\\ \quad of DFT vectors from $\textbf{F}_{i^*r}$ which has the maximum density (Each DFT vector is given\\ \quad a weight of $1$ if it is not covered in $A$ ). \\ 
		
			$4. $ $A \longleftarrow  A \oplus ([i^*,r],\textbf{F}_{i^*r})   $\\
		
			$5. $ \textbf{For} $u \notin A$ \textbf{do}:\\
		
			$6. $ \textbf{If} $ W (A  \oplus ([u,r],\textbf{F}_{ur}))\leq U^{[k]}_{r}$ \textbf{and} $ \forall \;\textbf{f} \in \textbf{F}_{ur},  V_{A \oplus ([u,r],\textbf{F}_{ur})} ([u,r],\textbf{f})> 0$ \textbf{do}:\\
		
			$7. $ $A \longleftarrow  A \oplus ([u,r],\textbf{F}_{ur})   $\\
	
			$8. $ \textbf{End for}\\
		
			$9. $ $j \longleftarrow  j+1  $\\
	
			$10. $ \textbf{End While}\\
		
			$11. $ \textbf{Return}($A$)\\
			\hline
			\hline
		\end{tabular}
		\captionof{table}{\textbf{Greedy(S,r,k)}\label{greedy}}
	\end{center}
		
		The algorithm  consists of solving a generalized maximum coverage problem  for each copilot group $C^{[b]}_k,k=1,\ldots,\tau, b=1,\ldots,N_c$, successively.  The main idea of the algorithm  is to use two  nested greedy phases. In the upper phase, the  maximum coverages $C^{[b]}_{k}, k=1..\tau, b=1,..,N_c$ are computed successively in a greedy manner. In order  to obtain each coverage, in the  lower phase, the  algorithm  uses the  residual  value  (\ref{eq:residual_def}) in a greedy procedure so as to choose  a  subset of DFT beams  that  are part of  a  given  user  spatial signature  with  the highest  density. Users are then added  to the selection  as long as the residual  value of their associated DFT vectors  is positive and the pilot reuse constraint is not violated.  When the  greedy  phase ends,  its resulting selection  is compared with the highest value  that can result from selecting a single  user.  The selection  with the best coverage is then returned. This procedure  is  repeated in the each cell $b=1,..,N_c$, for each copilot group $k=1..\tau$. We now derive the performance guarantee of the proposed algorithm.
	\begin{Theorem}\label{spatial_multiplexing:theor:ratio_knowledge}
		The  proposed algorithm in tables  ~\ref{aware} and ~\ref{greedy}  provides an $(1-(\frac{\tau-1}{\tau}) ^{\tau} )\frac{\frac{3}{2}- \frac{e^{-2}}{2}}{1- e^{-2}}$-approximation of the optimal solution of problem~\eqref{eq:problem_knowledge}.
	\end{Theorem}
	
	\begin{proof}
		See \ref{ratio_knowledge}.
	\end{proof}
	The  proposed algorithm in  table~\ref{aware} provide $\tau$ covers of  the DFT matrix beams in  each cell. This results in  $\tau$ copilot user groups, in each cell, that fully exploit all  available DoFs with minimum  overlapping between the  beam sets of each user. This leads to an efficient reduction of intra-cell copilot interference. The main differences between the  power agnostic and  aware cases are  performance  guarantees and complexity. Indeed, the  simplified  power agnostic case enables to  achieve the desired grouping  with  good performance guarantee (see Theorem \ref{spatial_multiplexing:theor:ratio_coarse}) and  low complexity. The  power aware case provide a more efficient  grouping since  it  takes into consideration the  users channel gains. Nevertheless, this  comes with  a penalty of a higher computational  complexity.
	Constructing the copilot user groups enables to  efficiently address the  intra-cell interference. Nevertheless, further SE  gain can be achieved by addressing the issue of  inter-cell copilot interference. This will be the focus of the next section. 

\section{Cross Cell Pilot allocation: a graphical approach }

\subsection{Cross Cell Pilot allocation problem}

     In this paper, we focus on mitigating copilot interference. Constructing copilot groups across cells,  in order to address the problems of intra-cell and inter-cell interference simultaneously, proves to be quite complex. This is due to the complexity of defining a proper grouping metric that is based on the spatial signatures of both useful and interference links.
     This fact motivated the present approach of dealing with interference through two consecutive subproblems addressing  intra and inter-cell copilot interference, respectively. 	A major advantage of such division is the reduction of complexity. Indeed, since copilot user groups have already  been constructed in each cell, addressing out-of-cell copilot interference reduces to allocating specific training sequences to copilot groups. If copilot interference is to be addressed from the beginning as a whole, it will result in a complex pilot allocation problem among all users in all cells which proves to  be complicated.	
	Practically, complete removal of interference is not physically possible. In addition,  copilot user grouping was performed with the clear goal of managing intra-cell interference. Consequently, when dealing with inter-cell copilot interference, previously formed copilot groups  should be maintained.   We propose a scheme in which  pilot allocation is done such that high interference links are suppressed when spatial  signature based receivers are used. In this section, we address this problem using an intelligent pilot assignment scheme.  The basic idea is to infer inter-cell copilot interference from the spatial  signatures of interference links. A training phase to obtain spatial  information of the interference links is  required. This can be implemented  without a large signaling overhead owing to  the slow changing  spatial  information.    We now consider that,  each user $i,b, i=1...K, \; b=1...N_c$ is associated with $N_c$ matrices $\{\textbf{F}^{[r]}_{ib}, r= 1...N_c\}$. Each matrix $\textbf{F}^{[r]}_{ib}$ is constructed in a similar manner to~\eqref{eq:projection_matrix} as follows

	\begin{align}\label{eq:interference_links}
	&   \textbf{F}^{[r]}_{ib}= \{ \textbf{f}_{s} \in \textbf{F},     \frac{\left\|  \textbf{g}^{[r]^\dagger}_{ ib} \textbf{f}_{s}\right\|^2 }{{Tr(\textbf{R}^{[r]}_{ib})}} \geq \alpha \},
	\end{align}

\subsection{Graphical Modeling and proposed solution}

		The first step to  manage inter-cell copilot interference is to construct an interference graph that corresponds to the considered system setting. 
		We construct  an undirected interference graph $\mathcal{G}(\mathcal{C},\mathcal{E})$. Each node $C^{[b]}_k \in \mathcal{C}, \;k=1...\tau,b=1...N_c$ represents  
		copilot user group of index $k $ in a give cell $b$.
		Each edge in $e_{C^{[j]}_{b},C^{[k]}_{l}} \in \mathcal{E}$ represents an interference  link and is associated with a given  weight $w_{C^{[j]}_{b},C^{[k]}_{l}}$. We propose a method for determining the edge weight without accurate SINR measurements since it cannot be obtained before pilot allocation. 
		In this work, we propose to infer interference levels form spatial information. This consideration is due to  the practical  low signaling overhead   that is required.
		
		Since the weight  of each edge quantifies the level  of interference between two copilot groups, an appropriate measure needs to be considered. This task is not an easy one since the weight of the edge between two different copilot groups should properly characterize the levels of resulting interference. The research papers that investigated spatial  division multiplexing proposed different metrics to characterize subspace distances. The most used one is chordal distance. Such metric has  the considerable downside of neglecting the actual signal power in each subspace. In addition, the present framework implies that the weight of each edge need to characterize the mutual interference between two groups of users. Consequently, defining a distance measure between copilot groups is a major issue in our case.
		In order to solve this issue, we call upon hierarchical  clustering  where   measuring distances between groups is commonly encountered. We adopt  a   linkage method in hierarchical  clustering \cite{tot}, namely weighted average linkage.
		To  quantify interference on each link,  we use  spatial signature overlapping between the users forming the  two copilot groups which is  obtained using the  chordal distance between spatial  signatures. The weight of each edge $w_{C^{[j]}_{b},C^{[k]}_{l}}$ is then given by
		\begin{align}\label{eq:weight}
		&  w_{C^{[j]}_{b},C^{[k]}_{l}}= \underset{y \in C^{[j]}_{b},z \in C^{[k]}_{l}}{\text{min}} \{\frac{1}{2  }  \lVert  \textbf{F}^{[b]}_{yb} \textbf{F}^{[b]^\dagger}_{yb} -    \textbf{F}^{[b]}_{zl} \textbf{F}^{[b]^\dagger}_{zl}   \rVert_F  ^2+ \frac{1}{2  }  \lVert \textbf{F}^{[l]}_{yb} \textbf{F}^{[l]^\dagger}_{yb} -    \textbf{F}^{[l]}_{zl} \textbf{F}^{[l]^\dagger}_{zl}  \rVert_F  ^2\}    ,
		\end{align}
		Here $\lVert  \textbf{F}^{[b]}_{yb} \textbf{F}^{[b]^\dagger}_{yb} -    \textbf{F}^{[b]}_{zl} \textbf{F}^{[b]^\dagger}_{zl}   \rVert_F  ^2$ represents the chordal distance between the  spatial  signatures of the  useful signal  of  user  $y \in C^{[j]}_{b}$ and the interference  generated by $z \in C^{[k]}_{l}$.
		$\lVert \textbf{F}^{[l]}_{yb} \textbf{F}^{[l]^\dagger}_{yb} -    \textbf{F}^{[l]}_{zl} \textbf{F}^{[l]^\dagger}_{zl}  \rVert_F  ^2$ denotes the chordal distance between the  spatial  signatures of the  useful signal  of  user  $z \in C^{[k]}_{l}$  and the interference  generated by $y \in C^{[j]}_{b}$.

		The weight expression~\eqref{eq:weight} captures the minimum  chordal distance between  the  spatial signatures of the  interference and useful signals for all users in the two  copilot groups and is inspired by the single  and weighted average linkage, commonly used in hierarchical  clustering \cite{tot}. In each  cell,  users from the same  copilot  group  are  the only devices allowed to  transmit the same UL training sequence. Consequently, during pilot allocation, we  need to make sure that  any given pilot  sequence $\textbf{q}_l, l=1...\tau$ should be allocated to  only one copilot group in  each  cell. In order to  do  so, the weight  of the links between copilot groups from the same  cell  will be given a very large	value $w_{\infty}$, because  intra-cell interference between copilot groups must be avoided. Using this metric we are able to  construct the  interference graph $\mathcal{G}(\mathcal{C},\mathcal{E})$, which is a first step in the  proposed pilot allocation  scheme. An illustration  of  $\mathcal{G}$ is presented in figure $2$ for the  case of $N_c=2$ and  $\tau= 2$.

	\begin{figure}[h!]
		\centering
		\includegraphics[width=14cm,height=8cm]{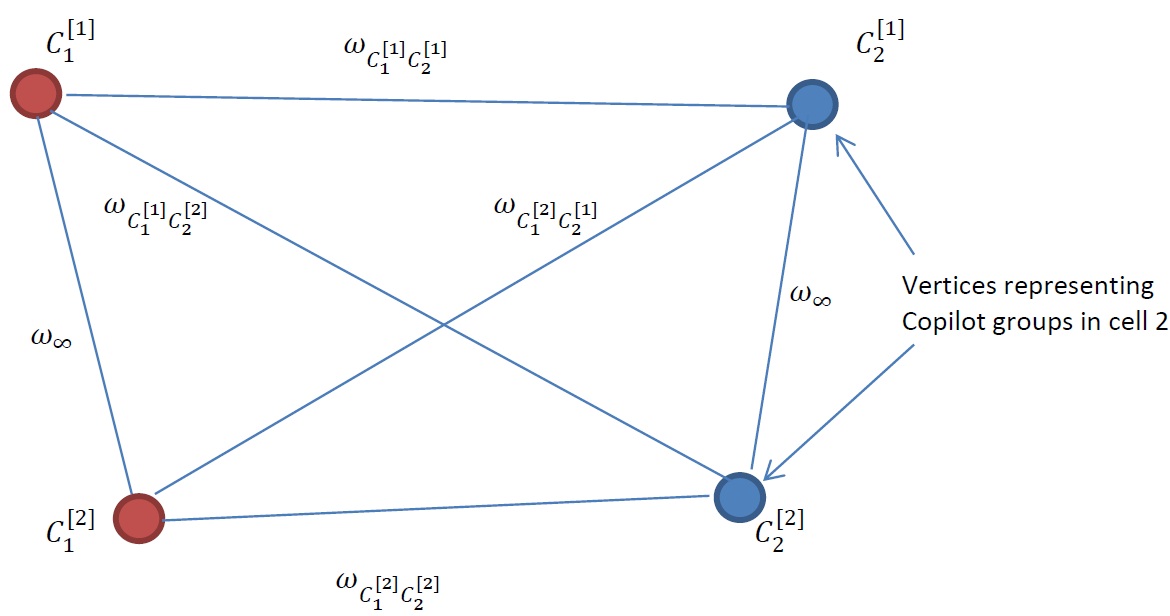}
		\label{System Model}
		\caption{ Interference  graph example }
	\end{figure}

	The considered  problem is  closely related to  MAX-CUT problem \cite{maxk}. Indeed, the task of interference management  reduces to suppressing high pilot contamination between copilot groups. This can be performed by allocating the same training sequence to copilot groups with minimum mutual interference weights. In the considered graphical framework, this task is equivalent to partitioning the interference graph into $\tau$ subgraphs where the copilot groups in each subgraph will be allocated the same training sequence.
	In the graph theory, a cut is a partition of the vertices of the graph into multiple sets. In  weighted graphs, the size of the cut
	is the sum of weights of the edges that cross the cut. A cut is said to be maximal if the size of the cut is not smaller than the size of any other cut. By generalizing this notion to $\tau$ cuts, the MAX-$\tau$-CUT problem is to find a set of $\tau$ cuts that are not smaller in size than any other $\tau$ cuts. Consequently,
	pilot allocation is equivalent to a MAX-$\tau$-CUT problem on the interference graph and can be stated as follows:\\
	\textbf{Pilot sequence allocation problem:} Given the  interference  graph $\mathcal{G}(\mathcal{C},\mathcal{E})$ with  $\tau \times N_c$ nodes and edge weight $w_{C^{[j]}_{b},C^{[k]}_{l}}$ for each  edge $e_{C^{[j]}_{b},C^{[k]}_{l}} \in \mathcal{E}$, partition the  graph into  $\tau$ disjoint sets $P_g, g=1,...,\tau$, such that\\ $ \sum\limits_{\substack{C^{[j]}_{b} \in P_g ,C^{[k]}_{l} \in P_{g'}\\ g \neq g' }}^{}     w_{C^{[j]}_{b},C^{[k]}_{l}}  $ is maximized.\\
	The training sequence  length  constraint  is  already  taken  into  consideration by  the definition of the  number of resulting  sets $\tau$. Since interference links between copilot  users in the  same cell  was assigned  a large weight $w_{\infty}$, we are sure that  all copilot groups within a given  cell  will be  allocated to  different  sets. 
	The max-$\tau$- cut algorithm assigns different  training sequences to copilot groups with  strong spatial  signatures overlapping between the useful and interference signals. The complexity of the proposed pilot allocation algorithm  depends on the  number of  copilot  groups, edges and training sequences. 
	
	A remark on  the complexity of this algorithm  is now in order. Proceeding to sequence allocation, once copilot  groups  are formed, results in a substantial  simplification  of  the problem. Instead of  processing  each user individually, the  proposed method  exploits the formed  copilot  groups in order to  reduce running time of the pilot allocation algorithm. This impact becomes very interesting in an  IoT communication  scenario. 
	
	The  Pilot sequence allocation problem is NP-hard \cite{Korte_comb}, meaning that  the  optimal  solution  is  computationally prohibitive to  obtain. Consequently, we  use the low complexity  algorithm in \cite{maxk}. The heuristic algorithm, in \cite{maxk}, provides an  approximate solution  that  achieves at  a ratio  of $(1- \frac{1}{\tau})$ of  the  optimal  one for a general  MAX-$\tau$-CUT problem, given that all weights in the graph are nonnegative integers.  
	Since all  weights in the considered interference  graph  are positive, using the  heuristic from  \cite{maxk} provides us with  a $(1- \frac{1}{\tau})$-approximation of the optimal  solution for the considered problem.
	The detailed algorithm for cross cell Pilot assignment is given in  table~\ref{cut}.

\begin{center}
	\begin{tabular}{ l  }
		\hline
		\hline
		\textbf{Initialize:}  intra-set weights $W_{g}=0, \forall g=1...\tau$
		$P_g=\emptyset, g=1,...,\tau$ \\
		$1. $ Assign the $\tau$ copilot groups in cell $1$ to different  pilot sets \\
		$2. $ Randomly order the rest of copilot groups.\\
		$3. $ Select the  next copilot group $v$ and assign it to  set  $g^*$ for which\\  $W^v_{g^*}$ is minimized where  $W^v_{g}= \sum_{u \in P_g}^{} w_{v,u} $  \\ 
		$4. $ Update the Average weight of group $g^*$ such that $W_{g^*}=W_{g^*}+W^v_{g^*}$\\
		$5. $ Repeat steps  $3-4$ until all  copilot groups are assigned.\\
		\hline
		\hline
	\end{tabular}
	\captionof{table}{\textbf{Cross cell Pilot assignment algorithm}\label{cut}}
\end{center}

	\section{Numerical Results AND Discussion}
	
	
In this section, we provide numerical  results demonstrating the performance of the proposed spatial basis coverage  copilot user selection. We  compare the proposed approach with  a conventional  TDD massive MIMO  system where all scheduled, within  each cell, are given orthogonal  training sequences. We then extend the simulation  results to include MAX-$\tau$-CUT pilot allocation. We  consider a network constituted of $N_c=4$  hexagonal  cells.  Each cell has a radius $0.5\;\text{Km}$ from center to vertex. Each cell contain  a massive MIMO  BS at its center, equipped with  $M=128$ equally spaced isotropic antennas. The minimum  distance  between  antenna elements is equal to  $\frac{\lambda}{2}$. Each cell  contains $K=25$ users with randomly generated mean direction of  arrivals. 
The channel vectors of the  different users are generated according to (\ref{eq:channel_decomp_spatial}) where $P=100$. Each  coefficient $ \mu^{[r]^2}_{ib}, \forall i,b,r$  denotes the  path-loss between the user and the target BS. The  path-loss coefficient is  $3.5$. 
For each user $i,b$, the angles of its rays   $\theta^{[r,p]}_{ib}, p=1,\ldots,P$ are uniformly distributed in the interval 
$\big[\theta^{[r]}_{ib}-\Delta\theta^{[r]}_{ib} ,\theta^{[r]}_{ib}+\Delta\theta^{[r]}_{ib}   \big]$ where the AS  is supposed to be the same for all users with $\Delta\theta^{[r]}_{ib}=\Delta= 4^{\circ}$. The coherence interval is set to $T_s=128$ samples, split between training and data transmission. We take $\alpha=0.05$.
In order to  assess the accuracy of channel  estimation,  we take as metric the average individual  mean square error (MSE).

		\begin{figure}[!htb]
			\centering	
			\includegraphics[scale=.6]{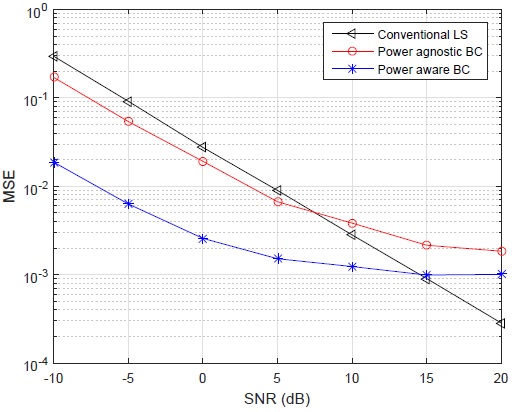}
			\captionsetup{font=footnotesize}
			\caption{Comparison of uplink channel  estimation  MSE  with  $\tau=8$ and  $U^{[k]}_{b}=3, \forall k, b$ \label{MSE}}
		\end{figure}
	
Figure (\ref{MSE}) illustrates a comparison of  {MSE}  performance  of {UL}  channel estimation, as a function of transmit signal to noise ratio
(TX SNR). Figure (\ref{MSE}) shows that  {UL}  channel estimation is improved when using the two proposed spatial basis coverage algorithms in the  low {SNR}  range (up to  approximately $7.5 $dB for the power agnostic  approach and  up to $15$ dB  for  the power aware approach  ). It shows that power aware spatial basis coverage outperform the power agnostic  approach in MSE. This is mainly due to the fact that the power aware approach  takes into consideration  the  channel  gain when forming copilot groups. In addition, the power aware approach  allow for less overlapping in spatial signature when compared to the power agnostic one. 	Figure  (\ref{MSE}) shows also that, as {SNR}  increases, the   performances of the two proposed algorithms reach two distinct error floors. This phenomenon is due to  the truncation error that  results from projection on the users specific spatial subspaces. These error floors depend  on the  rank of the  user's spatial signatures and can be  reduced by 
decreasing  $\alpha $ (i.e.  considering the DFT vectors that concentrate lower levels  of  the  user  channel  power, see \ref{eq:projection_matrix}). Note that, in the conventional  approach, all scheduled users, in  each cell, are allocated orthogonal training  sequences. Consequently, the conventional  CSI estimation uses $U^{[k]}_{b}\times \tau$ orthogonal pilot sequences. On the  other hand, the  proposed algorithms  use training  sequences of  length $\tau$. This explains, in part, MSE performance  in the  high SNR range where the conventional CSI estimation  approach can outperform the proposed  algorithms. However, our method enables to reduce the needed UL training resources which will result, ultimately,  in improving  SE.

 Nevertheless, we should emphasize on the fact that  the presented performances are attained with UL training  overheads of $\tau$ and $U^{[k]}_{b}\times \tau$,  for the  proposed algorithms and  the conventional approach respectively. This means  that, for a reasonable SNR range,  we are able to improve CSI estimation accuracy while reducing  the  UL training  overhead.

		\begin{figure}[!htb]
			\centering	
			\includegraphics[scale=.6]{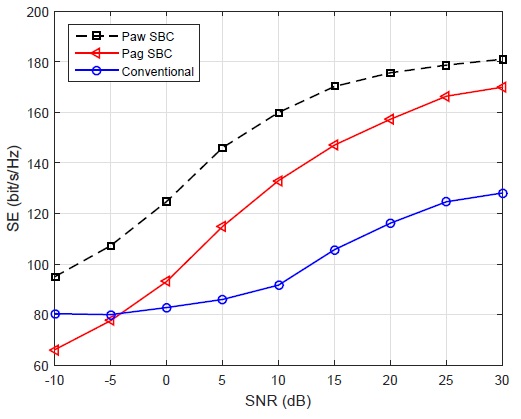}
			\captionsetup{font=footnotesize}
			\caption{Comparison of the achievable average SE  with  $\tau=10$ and  $U^{[k]}_{b}=2, \forall k, b$    \label{diff_snr}}
		\end{figure}

Figure (\ref{diff_snr}) illustrates a comparison of the  achievable SE  between  the proposed spatial basis coverage algorithms and a conventional TDD massive MIMO system, as a function of the  SNR. 	Figure (\ref{diff_snr}) shows that, for an SNR  of $0$ dB,  the power aware and the  power agnostic spatial basis coverage approaches achieve $124.75$ bits/Hz/s and $93.095$ bits/Hz/s respectively. This represent gains of $ 41.983 $bits/Hz/s  and $ 10.328 $ bits/Hz/s, respectively, in comparison with a  conventional  TDD massive MIMO system. As SNR increases the  gain in  spectral  efficiency  that  the proposed spatial basis coverage approaches provide also increases. This is mainly due to  the reduced impact of  additive noise since the  system becomes interference limited which emphasizes the ability of the proposed schemes to  mitigate intra-cell copilot interference.

							\begin{figure}[!htb]
								\centering	
								\includegraphics[scale=.6]{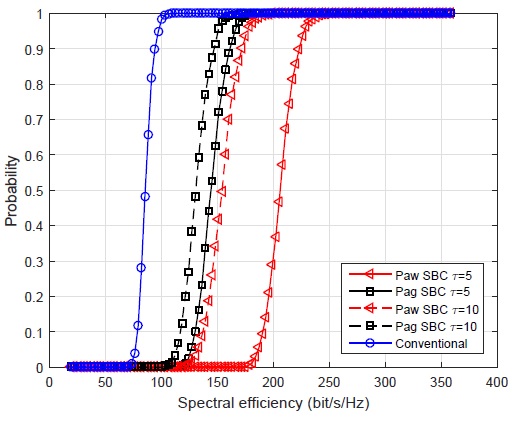}
								\captionsetup{font=footnotesize}
								\caption{Comparison of CDFs of achievable SE for  different  $\tau $ and  $U^{[k]}_{b}$ values with $SNR =10 $ dB   \label{reuse}}
							\end{figure}
		
		Figure (\ref{reuse}) illustrates a comparison of CDFs of the  achievable spectral efficiency  between  the proposed beam coverage algorithms and a conventional TDD mMIMO system, for different  for  different  $\tau $ and  $U^{[k]}_{b}$ values.	The proposed algorithms are applied in the  same  network setting   with $\tau=10, U^{[k]}_{b}=2 \; \forall k,b$ and  $\tau=5, U^{[k]}_{b}=4  \;\forall k,b$, respectively.
		Figure (\ref{reuse}) shows that, for  $\tau=10 $ and  $U^{[k]}_{b}=2 \; \forall k,b$, the power aware  and the  power agnostic spatial basis coverage approaches achieves $5\%$ outage rate around  $132$ bit/s/Hz and  $113$ bit/s/Hz, respectively.  For  $\tau=5 $ and  $U^{[k]}_{b}=4 \; \forall k,b$, $5\%$ outage rate is attained around $186$ bit/s/Hz and $126$ bit/s/Hz, respectively. This gain in performance  is  mainly  due  to the  spared training resources when decreasing  $\tau $ from $10$ to $5$. Figure (\ref{reuse})  shows that  the  impact of increased intra-cell  copilot interference, due to more pilot reuse,  can be efficiently  mitigated through  spatial basis  coverage and the speared training resources.

		\begin{figure}[!htb]
			\centering	
			\includegraphics[scale=.6]{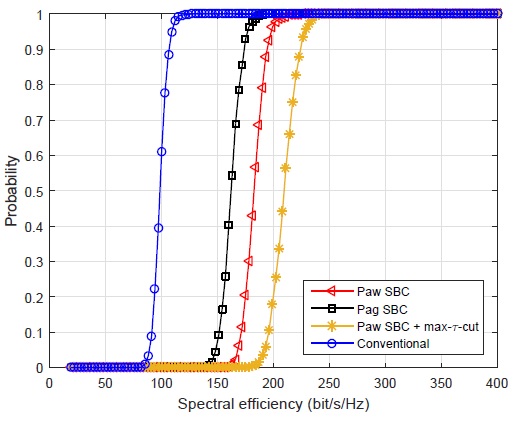}
			\captionsetup{font=footnotesize}
			\caption{Comparison of CDFs of achievable SE for  $\tau=4$ ,  $U^{[k]}_{b}=5$ and  with $SNR =10 $ dB   \label{cut_fig} }
		\end{figure}
		
Figure (\ref{cut_fig} ) illustrates the impact of the proposed max-$\tau$-cut pilot  allocation  algorithm.  	Figure (\ref{cut_fig} ) shows that addressing the  issue of  inter-cell copilot interference  through efficient  pilot sequence allocation  results in  an improvement in the system SE. Indeed, while  the  power aware spatial  basis cover approach  achieves $5\%$ outage rate around  $168$ bit/s/Hz,  the combination with  the  max-$\tau$-cut pilot assignment algorithm   achieves $5\%$ outage rate around  $191$ bit/s/Hz.
Consequently, after constructing copilot user groups based on the  spatial  basis  approach, the  same  diversity  in  spatial signatures can  be leveraged in order to  address the  problem of  inter-cell copilot interference. Although complete  removal of interference  is still  not possible, especially since  copilot groups are  constructed based on the useful  links  spatial  information, non-negligible  performance improvement can  be  achieved  by efficient  pilot sequence  allocation  across  cells.

	\section{Conclusion }
	In this paper, we have studied user scheduling and pilot allocation based on spatial  division multiplexing  for TDD massive MIMO systems. We proposed a copilot user grouping approach based on DFT basis coverage. 
	After associating each user with the DFT  vectors that  concentrate the majority  of  its  channel  power, users are assigned to  copilot groups in order to achieve a  maximum coverage of  the DFT  vectors per group with minimum  overlapping in their  respective spatial  signatures.
	The proposed approach  enables to  increase the spectral efficiency  while reducing the required training overhead. This is achieved by enabling accurate CSI estimation leveraging  spatial diversity. Various numerical results were provided to demonstrate the effectiveness  of the proposed	grouping approach.
	In order to efficiently manage inter-cell copilot interference,  further optimization  is performed. We have proposed a  graphical approach for training sequence allocation  across cells. Leveraging the  interference links spatial information and the  previously constructed copilot groups,  training sequence allocation is formulated as  a max-cut problem. Thus enabling the utilization of  a low complexity  algorithm for training sequences allocation. 
	Although  the proposed approach does not remove entirely  copilot interference,  it provides a practical method to meet essential  requirements  for $5$G and beyond networks, namely,  increasing  spectral  efficiency and  connection  density for the  same  training overhead.

	\section{Appendix}
	\textbf{Appendix A proof of Theorem 1:}\label{ratio_coarse}\\     
		In order to derive a tight bound on the performance of the proposed algorithm in table $I$,  we consider the worst case behavior of the greedy  heuristic as in \cite{subopt}. However  in our case, the  problem  is more  sophisticated  since  it  includes the  combination  of the  greedy  heuristic with  an  approximation  algorithm.
		
			The optimization problem~\eqref{eq:problem_coarse} can be decomposed into $N_c$ independent problems, each of which is defined in a given cell $b=1,\ldots,N_c$. We start by  defining $C^{[b]}_k,k=1,\ldots,\tau, b=1,\ldots,N_c$ as the maximum coverage at iteration $k$ of the approximate algorithm in table $I$. We define  $C^{[b]}_{k_{opt}},k=1,\ldots,\tau, b=1,\ldots,N_c$ as the optimal maximum coverage that can be obtained at iteration $k$.
		We also consider $C^{[b]}_{opt}$ as the optimal solution the problem of~\eqref{eq:problem_coarse} defined in each cell $b=1,\ldots,N_c$:
		\begin{align}\label{eq:problem_coarse_percell}
		V(C^{[b]}_{opt})=\underset{Y}{\text{max}} & \sum_{k=1}^{\tau}   \sum_{s=1}^{M}  y^{[k]}_{s,b}\\\tag{15a}
		\text{s.t} \;  & \sum_{i}^{}   x^{[k]}_{\{i,b\}}  \leq U^{[k]}_{b} \;\; \forall k=1...\tau\\\tag{15b}
		& \sum_{i, f_s \in F_{ib} }^{} x^{[k]}_{\{i,b\}}  \geq y^{[k]}_{s,b}  \;\; \forall k=1...\tau, \nonumber
		\end{align}
		where $	V(C^{[b]}_{opt})$ represents the value of  the  coverage $C^{[b]}_{opt}, \; b=1,\ldots,N_c$. 
		The objective function in (\ref{eq:problem_coarse_percell}) is modular. Consequently, the following property holds $\forall b=1,\ldots,N_c$:
		
		\begin{align}\label{eq:modular}
		& V( C^{[b]}_{opt} ) \leq   \sum_{k=1}^{t-1} V(C^{[b]}_{k_{opt}})  + \tau V(C^{[b]}_{t_{opt}})   \;\;,\forall t= 1,..., \tau,
		\end{align}
		In order to  derive a bound on the achievable performance of the proposed algorithm, we consider the worst case behavior of the greedy heuristic. Doing so is equivalent to solving the following linear problems $\forall ~b=1,...,N_c$:
		\begin{align}\label{eq:linear}
		P(b)= {\text{min}} & \sum_{k=1}^{j}  \frac{V(C^{[b]}_{k_{opt}})}{V(C^{[b]}_{opt})}\\\tag{16a}
		\text{s.t} \;  & \sum_{k=1}^{t-1}  \frac{V(C^{[b]}_{k_{opt}})}{V(C^{[b]}_{opt})} + \tau \frac{V(C^{[b]}_{t_{opt}})}{V(C^{[b]}_{opt})} \geq 1 \;\; \forall t=1...j \\\nonumber 
		\end{align}
		Where the constraint $(16a)$ is obtained from~\eqref{eq:modular}. Since ~\eqref{eq:linear} is  a linear  problem, we can  solve it  by  considering its dual $\forall b=1,...,N_c$ which can be written as follows
		\begin{align}\label{eq:dual}
		D(b)={\text{max}} & \sum_{t=1}^{j+1} v_t\\\tag{17a}
		\text{s.t} \;  &  \tau v_k +         \sum_{t=k+1}^{j+1} v_t = 1 \;\;, \forall k=1...j\\\nonumber \tag{17b}
		& v_t \geq 0 \;\;\forall t=1...j+1 \\\nonumber
		\end{align}
		We now proceed by solving~\eqref{eq:dual}, and, consequently, by linear programming duality,~\eqref{eq:linear}. 
		Let $\upsilon= 1-v_{j+1} $, where $v_{j+1}$ is defined  in~\eqref{eq:dual}. Consequently, $\;\forall k=1...j$,  we have:
		\begin{align}\label{eq:lambda}
		& \tau v_k +         \sum_{t=k+1}^{j} v_t = \upsilon ~\text{and}~ v_k= \frac{\upsilon-  \sum_{t=k+1}^{j} v_t}{\tau}
		\end{align}
		Then  $v_{k}$ are  calculated iteratively $\forall\; k=1...j$. 	Indeed, 		
		\begin{align}\label{eq:lambda_v}
		  v_j= \frac{\upsilon}{\tau} ,\; v_{j-1}= \frac{\upsilon- v_j}{\tau}= \frac{\upsilon- \frac{\upsilon}{\tau}}{\tau}=\frac{\upsilon}{\tau} (\frac{\tau-1}{\tau}),\ldots, \; v_1= \frac{\upsilon-  \sum_{t=2}^{j} v_t}{\tau} =  \frac{\upsilon}{\tau} (\frac{\tau-1}{\tau})^{j-1}.
		\end{align}		
		Consequently, $v_{k}, \forall\; k=1...j$ are given by
		\begin{align}\label{eq:lambda2}
		&  v_k= \frac{\upsilon}{\tau} (\frac{\tau-1}{\tau})^{j-k}
		\end{align}
		Therefore, finding $D(b), \forall b=1,..., N_c$ is equivalent to the following :
		\begin{align}\label{eq:lambda3}
		D(b)= \underset{ 0 \leq \upsilon \leq 1}{\text{max}} (\sum_{t=1}^{j+1} v_t)      =\underset{ 0 \leq \upsilon \leq 1}{\text{max}} (\upsilon(1-(\frac{\tau-1}{\tau})^{j}))
		\end{align}
		which is achieved for  $\upsilon=1$. 
		It follows that, $ P(b) = 1 - \frac{j}{\tau} (\frac{\tau-1}{\tau}) ^{j}$. Taking $j=\tau$, we obtain: 
		\begin{align}\label{eq:lambda4}
		P(b) = 1 - (\frac{\tau-1}{\tau}) ^{\tau} 
		\end{align}
		Consequently, the obtained solution using a greedy sequential maximum  coverage verifies:
		\begin{align}\label{eq:lambda5}
		&\frac{V(C^{[b]}_{opt})-\sum_{k=1}^{\tau} {V(C^{[b]}_{k_{opt}})}}{V(C^{[b]}_{opt})} \leq (\frac{\tau-1}{\tau}) ^{\tau}  
		\end{align}		
		In order to derive the performance guarantee of the  algorithm in table $I$, the approximation ratio of the used algorithm to  perform maximum  coverage, at each iteration, needs to be accounted for. We use the  following result which is based on the approximation ratio  given in \cite{MCP}. 
		\begin{Lemma}
			For any given cell index $b=1,\ldots,N_c$ and iteration $k=1,\ldots,\tau$, the implemented algorithm  provides a  $(1-\frac{1}{e})$ approximation of the  optimal cover, i.e   	 $C^{[b]}_k \geq  (1-\frac{1}{e}) C^{[b]}_{k_{opt}} ,k=1,\ldots,\tau, b=1,\ldots,N_c   $.
		\end{Lemma}  
		Since at each iteration of the algorithm, in table $I$,  we obtain a $(1-\frac{1}{e})$ approximation of the optimal maximum  coverage, we obtain the following for $ b=1,\ldots,N_c   $
			\begin{align}\label{eq:lambda6}
			&\sum_{k=1}^{\tau} {V(C^{[b]}_{k})} \geq (1-\frac{1}{e}) \sum_{k=1}^{\tau} {V(C^{[b]}_{k_{opt}})}\\
			&\sum_{k=1}^{\tau} {V(C^{[b]}_{k})} \geq (1-\frac{1}{e})(1-(\frac{\tau-1}{\tau}) ^{\tau} )V(C^{[b]}_{opt}) \nonumber
			\end{align}
		we can deduce that 	the algorithm in table $I$  provides a $(1-(\frac{\tau-1}{\tau}) ^{\tau} )(1-\frac{1}{e})$-approximation for each subproblem of~\eqref{eq:problem_coarse}. Taking the sum over $b=1,...,N_c$ finishes the proof.\\
		\textbf{Appendix B proof of Theorem 2:} \label{ratio_knowledge}\\     
		The proof for the performance bound of the  proposed  algorithm in table $II$ follows the same reasoning as the proof of Theorem $1$. The main idea is also to consider the worst case behavior of the greedy  heuristic with a change in the achievable approximation ratio at each iteration. 
		The optimization problem~\eqref{eq:problem_knowledge} can be decomposed into $N_c$ independent problems, each of which is defined in a given cell $b=1,\ldots,N_c$. We start by  defining $C^{[b]}_k,k=1,\ldots,\tau, b=1,\ldots,N_c$ as the maximum coverage at iteration $k$ of the algorithm in table $II$. We define  $C^{[b]}_{k_{opt}},k=1,\ldots,\tau, b=1,\ldots,N_c$ as the optimal maximum coverage that can be obtained at iteration $k$.
		We also consider $C^{[b]}_{opt}$ as the optimal solution the problem of~\eqref{eq:problem_knowledge} defined in each cell $b=1,\ldots,N_c$:
			\begin{align}\label{eq:problem_knowledge_percell}
		V(C^{[b]}_{opt})=	\underset{Y}{\text{max}} & \sum_{k=1}^{\tau} \sum_{i\in \Gamma(b)}^{}  \sum_{f_s\in \textbf{F}}^{}	\zeta^{[s]}_{ib} y^{[s,k]}_{\{i,b\}}     \\\tag{25a}
			\text{subject to} \;  & \sum_{i\in \Gamma(b), f_s \in F_{ib}}^{} y^{[s,k]}_{\{i,b\}}   \leq 1 \;\; \forall k=1...\tau,\\\tag{25b}
			& \sum_{i\in \Gamma(b), f_s \in F_{ib} }^{} x^{[k]}_{\{i,b\}}    \geq     y^{[s,k]}_{\{i,b\}}   \;\; \forall k=1...\tau\\\tag{25c}
			& \sum_{i\in \Gamma(b)}^{} x^{[k]}_{\{i,b\}}   \leq U^{[k]}_{b}   \;\; \forall k=1...\tau,\nonumber
			\end{align}

		The objective function in (\ref{eq:problem_knowledge_percell}) is modular. Consequently, the following property holds $\forall b=1,\ldots,N_c$:
		\begin{align}\label{eq:modularbis}
		& V( C^{[b]}_{opt} ) \leq   \sum_{k=1}^{t-1} V(C^{[b]}_{k_{opt}})  + \tau V(C^{[b]}_{t_{opt}})   \;\;,\forall t= 1,..., \tau,
		\end{align}
		In order to  derive a bound on the achievable performance of the proposed algorithm, we consider the worst case behavior of the greedy heuristic. Doing so is equivalent to solving the following linear problems $\forall ~b=1,...,N_c$:
		\begin{align}\label{eq:linear_knwoledge}
		P(b)= {\text{min}} & \sum_{k=1}^{j}  \frac{V(C^{[b]}_{k_{opt}})}{V(C^{[b]}_{opt})}\\\tag{27a}
		\text{s.t} \;  & \sum_{k=1}^{t-1}  \frac{V(C^{[b]}_{k_{opt}})}{V(C^{[b]}_{opt})} + \tau \frac{V(C^{[b]}_{t_{opt}})}{V(C^{[b]}_{opt})} \geq 1 \;\; \forall t=1...j \nonumber 
		\end{align}
		Where the constraint $(27a)$ is obtained from~\eqref{eq:modularbis}. Here also (\ref{eq:linear_knwoledge})  is solved using its dual. It follows that, $ P(b) = 1 - \frac{j}{\tau} (\frac{\tau-1}{\tau}) ^{j}$. Taking $j=\tau$, we obtain: 
		\begin{align}\label{eq:lambda4bis}
		P(b) = 1 - (\frac{\tau-1}{\tau}) ^{\tau}, \forall b=1,\ldots,N_c.
		\end{align}
		Consequently, the obtained solution using a greedy sequential maximum  coverage verifies:
		\begin{align}\label{eq:lambda5bis}
		&\frac{V(C^{[b]}_{opt})-\sum_{k=1}^{\tau} {V(C^{[b]}_{k_{opt}})}}{V(C^{[b]}_{opt})} \leq (\frac{\tau-1}{\tau}) ^{\tau}  
		\end{align}
		In order to derive the performance guarantee of the proposed algorithm  in table $II$, the approximation ratio of the used algorithm , at each iteration, needs to be considered. We use the  following result which is based on the approximation ratio  given in \cite{GMC}. 
		\begin{Lemma}
			For any given cell index $b=1,\ldots,N_c$ and iteration $k=1,\ldots,\tau$, using  an algorithm of approximation  ratio $\beta $ at step $3$ of the  algorithm in table $III$, the implemented algorithm  provides a  $  \frac{1+\beta-\beta e^{-\frac{1}{\beta}}}{1-e^{-\frac{1}{\beta}}} $ approximation of the  optimal cover, i.e   	 $C^{[b]}_k \geq  \frac{1+\beta-\beta e^{-\frac{1}{\beta}}}{1-e^{-\frac{1}{\beta}}} C^{[b]}_{k_{opt}} ,k=1,\ldots,\tau, b=1,\ldots,N_c   $.
		\end{Lemma}
		In this  work, we use the greedy  algorithm for  knapsack problems in step  $3$ of the  algorithm in table  $III$. Consequently, in our case, $\beta=\frac{1}{2} $ and  the approximation  ratio, at each iteration,  becomes $  \frac{\frac{3}{2}- \frac{e^{-2}}{2}}{1- e^{-2}}  $.\\	
		Since at each iteration of the algorithm, in table $II$,  we obtain a $  \frac{\frac{3}{2}- \frac{e^{-2}}{2}}{1- e^{-2}}  $ approximation of the optimal maximum  coverage, we obtain the following for $ b=1,\ldots,N_c   $
		\begin{align}\label{eq:lambda6bis}
		&\sum_{k=1}^{\tau} {V(C^{[b]}_{k})} \geq  \frac{\frac{3}{2}- \frac{e^{-2}}{2}}{1- e^{-2}}  \sum_{k=1}^{\tau} {V(C^{[b]}_{k_{opt}})}\\
		&\sum_{k=1}^{\tau} {V(C^{[b]}_{k})} \geq   \frac{\frac{3}{2}-\frac{e^{-2}}{2}}{1- e^{-2}}  (1-(\frac{\tau-1}{\tau}) ^{\tau} )V(C^{[b]}_{opt})\nonumber
		\end{align}
		we can deduce that 	the algorithm in table $II$  provides a $(1-(\frac{\tau-1}{\tau}) ^{\tau} ) \frac{\frac{3}{2}- \frac{e^{-2}}{2}}{1- e^{-2}}  $-approximation for each subproblem of~\eqref{eq:problem_knowledge}. Taking the sum over $b=1,...,N_c$ finishes the proof.

\end{document}